%% file: MVC.tex
\documentclass[preprint, nocopyrightspace]{sigplanconf}

\usepackage{graphicx}
\usepackage{color}  
\usepackage{algorithmic}
\usepackage{amssymb}
\usepackage{amsmath}
\usepackage{color}
\usepackage{times}

\input{macros}

\begin{document}

\title{Multiversion Conflict Notion}


\authorinfo{Priyanka Kumar}
	{priyanka@iitp.ac.in}
	{Indian Institute of Technology Patna, India}

\authorinfo{Sathya Peri}
	{sathya@iitp.ac.in}
	{Indian Institute of Technology Patna, India}


\maketitle              
\thispagestyle{empty}

\begin{abstract}
This work introduces the notion of multi-version conflict notion. 
\end{abstract}


\input{intro}

\input{model}

\input{mcn}

{
\bibliography{citations}
}

\end{document}

%% file: macros.tex
\newcommand{\punt}[1]{}
\newcommand{\cmnt}[1]{}

\newtheorem{theorem}{Theorem}

\newtheorem{lemma}[theorem]{Lemma}
\newtheorem{corollary}[theorem]{Corollary}

\newtheorem{definition}{Definition}

\newenvironment{proof}[1][Proof]{\noindent\textbf{#1.} }{} 

\newcommand{\figref}[1]{Figure~\ref{fig:#1}}

\newcommand{\lemref}[1]{Lemma~\ref{lem:#1}}

\newcommand{\eqnref}[1]{Eqn(\ref{eq:#1})}

\newcommand{\theqed}{$\Box$}

\newcommand{\qed}{\hspace*{\fill}\theqed\\\vspace*{-0.5em}}

\newcommand{\Wset}{\textit{Wset}}

\newcommand{\ignore}[1]{}

\newcommand {\comm} {committed}
\newcommand {\aborted} {aborted}
\newcommand{\tobj} {transaction object}

\newcommand{\txns}{txns}
\newcommand{\evts}[1] {evts(#1)}
\newcommand{\ssch} {sub-history}

\newcommand{\tseq} {t-sequential}

\newcommand{\lastw} {lastWrite}
\newcommand{\lwrite}[2] {#2.lastWrite(#1)}
\newcommand{\vwrite} {valWrite}
\newcommand{\valw}[2] {#2.valWrite(#1)}

\newcommand{\valid} {valid}

\newcommand{\legal} {legal}

\newcommand{\op} {operation}
\newcommand{\termop} {terminal operation}
\newcommand{\cc} {correctness-criterion}

\newcommand{\tryc} {tryC}
\newcommand{\trya} {tryA}

\newcommand{\opq} {opaque}
\newcommand{\opty} {opacity}
\newcommand{\coop} {co-opaque}
\newcommand{\coopty} {co-opacity}
\newcommand{\mvop} {mvc-opaque}
\newcommand{\mvopty} {mvc-opacity}

\newcommand{\mvsr} {MVSR}
\newcommand{\csr} {CSR}
\newcommand{\mvs} {multi-version STM}

\newcommand{\multv} {multi-versioned}
\newcommand{\mvc} {multi-version conflict}
\newcommand{\mvco} {mvc}

\newcommand{\mvcg}[1] {MVCG(#1)}
\newcommand{\mvg} {multi-version conflict graph}

%% file: intro.tex
\section{Introduction}
\label{sec:intro}

In recent years, Software Transactional Memory systems (STMs) \cite{HerlMoss:1993:SigArch,ShavTou:1995:PODC} have garnered significant interest as an elegant alternative for addressing concurrency issues in memory. STM systems take optimistic approach. Multiple transactions are allowed to execute concurrently. On completion, each transaction is validated and if any inconsistency is observed it is \emph{aborted}. Otherwise it is allowed to \emph{commit}. 

An important requirement of STM system is to ensure that transactions do not abort unnecessarily. This referred to as the \emph{progress} condition. It would be ideal to abort a transaction only when it does not violate correctness requirement (such as opacity). However it was observed in \cite{attiyaHill:sinmvperm:tcs:2012} that many STM systems developed so far spuriously abort transactions even when not required. A \emph{permissive} STM \cite{Guer+:disc:2008} does not abort a transaction unless committing of it violates the \cc{}.

With the increase in concurrency, more transactions may conflict and abort, especially in presence many long-running transactions which can have a very bad impact on performance \cite{AydAbd:2008:Serial:transact}. Perelman et al \cite{Perel+:2011:SMV:DISC} observe that read-only transactions play a significant role in various types of applications. But long read-only transactions could be aborted multiple times in many of the current STM systems \cite{herlihy+:2003:stm-dynamic:podc,dice:2006:tl2:disc}. In fact Perelman et al \cite{Perel+:2011:SMV:DISC} show that many STM systems waste 80\% their time in aborts due to read-only transactions. 

It was observed that by storing multiple versions of each object, \mvs{s} can ensure that more read \op{s} succeed, i.e., not return abort. History $H1$ illustrates this idea. $H1:r_1(x, 0) w_2(x, 10) w_2(y, 10) c_2 r_1(y, 0) c_1$. In this history the read on $y$ by $T_1$ returns 0 instead of the previous closest write of 10 by $T_2$. This is possible by having multiple versions for $y$. As a result, this history is \opq{} with the equivalent correct execution being $T_1 T_2$. Had there not been multiple versions, $r_2(y)$ would have been forced to read the only available version which is 10. This value would make the read cause $r_2(y)$ to not be consistent (\opq) and hence abort. 

\begin{figure}[tbph]
\centerline{\scalebox{0.5}{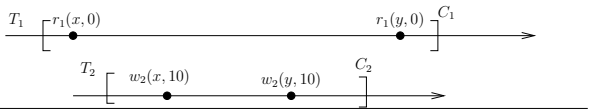}}
\caption{Pictorial representation of a History $H1$}
\label{fig:ex1}
\end{figure}

\cmnt {   
\paragraph{Roadmap.} 
The paper is organized as follows. We describe our system model in
Section~\ref{sec:model}. In Section~\ref{sec:mcn} we formally define
the graph characterization for implementing the \emph{Conflict opacity}.
In Section~\ref{sec:algo}, we describe the working principle of MCN algorithm. In Section~\ref{sec:gar} we are collecting the garbage. Our results are summarized in Section~\ref{sec:conc}.
}

%% file: ex1.pdf_t
\begin{picture}(0,0)%
\includegraphics{ex1.pdf}%
\end{picture}%
\setlength{\unitlength}{4144sp}%
\begingroup\makeatletter\ifx\SetFigFont\undefined%
\gdef\SetFigFont#1#2#3#4#5{%
  \reset@font\fontsize{#1}{#2pt}%
  \fontfamily{#3}\fontseries{#4}\fontshape{#5}%
  \selectfont}%
\fi\endgroup%
\begin{picture}(6369,1227)(1249,-703)
\put(5986,299){\makebox(0,0)[b]{\smash{{\SetFigFont{12}{14.4}{\familydefault}{\mddefault}{\updefault}{\color[rgb]{0,0,0}$r_1(y,0)$}%
}}}}
\put(6436,389){\makebox(0,0)[lb]{\smash{{\SetFigFont{12}{14.4}{\familydefault}{\mddefault}{\updefault}{\color[rgb]{0,0,0}$C_1$}%
}}}}
\put(4681,-421){\makebox(0,0)[b]{\smash{{\SetFigFont{12}{14.4}{\familydefault}{\mddefault}{\updefault}{\color[rgb]{0,0,0}$w_2(y,10)$}%
}}}}
\put(1396,299){\makebox(0,0)[b]{\smash{{\SetFigFont{12}{14.4}{\familydefault}{\mddefault}{\updefault}{\color[rgb]{0,0,0}$T_1$}%
}}}}
\put(2116,299){\makebox(0,0)[b]{\smash{{\SetFigFont{12}{14.4}{\familydefault}{\mddefault}{\updefault}{\color[rgb]{0,0,0}$r_1(x,0)$}%
}}}}
\put(5446,-241){\makebox(0,0)[lb]{\smash{{\SetFigFont{12}{14.4}{\familydefault}{\mddefault}{\updefault}{\color[rgb]{0,0,0}$C_2$}%
}}}}
\put(3106,-376){\makebox(0,0)[b]{\smash{{\SetFigFont{12}{14.4}{\familydefault}{\mddefault}{\updefault}{\color[rgb]{0,0,0}$w_2(x,10)$}%
}}}}
\put(2251,-286){\makebox(0,0)[b]{\smash{{\SetFigFont{12}{14.4}{\familydefault}{\mddefault}{\updefault}{\color[rgb]{0,0,0}$T_2$}%
}}}}
\end{picture}%

%% file: model.tex
\section{System Model and Preliminaries}
\label{sec:model}

The notions and definitions described in this section follow the definitions of \cite{KuzSat:Corr:2012}. We assume a system of $n$ processes, $p_1,\ldots,p_n$ that access a
collection of \emph{objects} via atomic \emph{transactions}.
The processes are provided with  four \emph{transactional operations}: the
\textit{write}$(x,v)$ operation that updates object $x$ with value
$v$, the \textit{read}$(x)$ operation that returns a value read in
$x$, \textit{tryC}$()$ that tries to commit the transaction and
returns \textit{commit} ($c$ for short) or \textit{abort} ($a$ for
short), and \textit{\trya}$()$ that aborts the transaction and returns
$A$. The objects accessed by the read and write \op{s} are called as
\tobj{s}. For the sake of simplicity, we assume that the values written by all the transactions are unique. 

Operations \textit{write}, \textit{read} and \textit{\tryc}$()$ may
return $a$, in which case we say that the operations \emph{forcefully
abort}. Otherwise, we say that the operation has \emph{successfully}
executed.  Each operation is equipped with a unique transaction
identifier. A transaction $T_i$ starts with the first operation and
completes when any of its operations returns $a$ or $c$. 
Abort and commit \op{s} are called \emph{\termop{s}}. 

For a transaction $T_k$, we denote all its  read \op{s} as $Rset(T_k)$
and write \op{s} $Wset(T_k)$. Collectively, we denote all the \op{s}
of a  transaction $T_i$ as $\evts{T_k}$. 

\vspace{1mm}
\noindent
\textit{Histories.} A \emph{history} is a sequence of \emph{events}, i.e., a sequence of
invocations and responses of transactional operations. The collection
of events is denoted as $\evts{H}$. For simplicity, we only consider
\emph{sequential} histories here: the invocation of each transactional
operation is immediately followed by a matching response. Therefore,
we treat each transactional operation as one atomic event, and let
$<_H$ denote the total order on the transactional operations incurred
by $H$. With this assumption the only relevant events of a transaction
$T_k$ are of the types: $r_k(x,v)$, $r_k(x,A)$, $w_k(x, v)$, $w_k(x,
v,A)$, $\tryc_k(C)$ (or $c_k$ for short), $\tryc_k(A)$, $\trya_k(A)$ (or $a_k$ for short). 
We identify a history $H$ as tuple $\langle \evts{H},<_H \rangle$. 

\cmnt {
\begin{figure}[tbph]
\centerline{\scalebox{0.5}{\input{figs/ex1.pdf_t}}}
\caption{Pictorial representation of a History $H1$}
\label{fig:ex1}
\end{figure}
}

Let $H|T$ denote the history consisting of events of $T$ in $H$, and $H|p_i$ denote the history consisting of events of $p_i$ in $H$. We only consider \emph{well-formed} histories here, i.e., (1) each $H|T$ consists of  a read-only prefix (consisting of read operations only), followed by a write-only part (consisting of write operations only), possibly \emph{completed} with a $\tryc$ or $\trya$ operation\footnote{This restriction brings no loss of  generality~\cite{KR:2011:OPODIS}.}, and (2) each $H|p_i$ consists of a sequence of transactions, where no new transaction begins before the last transaction completes (commits or a aborts). 

We assume that every history has an initial committed transaction $T_0$ that initializes all the data-objects with 0. The set of transactions that appear in $H$ is denoted by $\txns(H)$. The set of committed (resp., aborted) transactions in $H$ is denoted by $\comm(H)$ (resp., $\aborted(H)$). The set of \emph{incomplete} (or \emph{live}) transactions in $H$ is denoted by $\emph{incomplete}(H)$ ($\emph{incomplete}(H)=\txns(H)-\comm(H)-\aborted(H)$). For a history $H$, we construct the \emph{completion} of $H$, denoted $\overline{H}$, by inserting $a_k$ immediately after the last event of every transaction $T_k \in \emph{incomplete}(H)$.

\vspace{1mm}
\noindent
\textit{Transaction orders.} For two transactions $T_k,T_m \in \txns(H)$, we say that  $T_k$ \emph{precedes} $T_m$ in the \emph{real-time order} of $H$, denote $T_k\prec_H^{RT} T_m$, if $T_k$ is complete in $H$ and the last event of $T_k$ precedes the first event of $T_m$ in $H$. If neither $T_k\prec_H^{RT} T_m$ nor $T_m \prec_H^{RT} T_k$, then $T_k$ and $T_m$ \emph{overlap} in $H$. A history $H$ is \emph{\tseq{}} if there are no overlapping transactions in $H$, i.e., every two transactions are related by the real-time order. 

\cmnt {
\vspace{1mm}
\noindent
\textit{Sub-histories.} A \textit{sub-history}, $SH$ of a history
$H$ denoted as the tuple $\langle \evts{SH},$ $<_{SH}\rangle$ and is
defined as: (1) $<_{SH} \subseteq <_{H}$; (2) $\evts{SH} \subseteq
\evts{H}$; (3) If an event of a transaction $T_k\txns(H)$ is in $SH$ then all
the events of $T_k$ in $H$ should also be in $SH$. 
For a history
$H$, let $R$ be a subset of $txns(H)$, the transactions in
$H$. 
Then $\shist{R}{H}$ denotes  the \ssch{} of $H$ that is
formed  from the \op{s} in $R$. 
}

\vspace{1mm}
\noindent
\textit{Valid, legal and multi-versioned histories.} Let $H$ be a history and $r_k(x, v)$ be a successful read {\op} (i.e $v \neq A$) in $H$. Then $r_k(x, v)$, is said to be \emph{\valid} if there is a transaction $T_j$ in $H$ that commits before $r_K$ and $w_j(x, v)$ is in $\evts{T_j}$. Formally, $\langle r_k(x, v)$  is \valid{} $\Rightarrow \exists T_j: (c_j <_{H} r_k(x, v)) \land (w_j(x, v) \in \evts{T_j}) \land (v \neq A) \rangle$. We say that the commit \op{} $c_j$ is $r_k$'s \emph{\vwrite} and formally denote it as $\valw{r_k}{H}$. If there are multiple such committed transactions that write $v$ to $x$, then $r_k$ \vwrite{} is the commit \op{} closest to $r_x$. The history $H$ is \valid{} if all its successful read \op{s} are \valid. 


We define $r_k(x, v)$'s \textit{\lastw{}} as the latest commit event $c_i$ such that $c_i$ precedes $r_k(x, v)$ in $H$ and $x\in\Wset(T_i)$ ($T_i$ can also be $T_0$). Formally, we denote it as $\lwrite{r_k}{H}$. A successful read \op{} $r_k(x, v)$ (i.e $v \neq A$), is said to be \emph{\legal{}} if transaction $T_i$ (which contains  $r_k$'s \lastw{}) also writes $v$ onto $x$. Formally, $\langle r_k(x, v)$ \text{is \legal{}} $\Rightarrow (v \neq A) \land (\lwrite{r_k(x, v)}{H} = c_i) \land (w_i(x,v) \in \evts{T_i}) \rangle$.  The history $H$ is \legal{} if all its successful read \op{s} are \legal. Thus from the definitions we get that if $H$ is \legal{} then it is also \valid.

\figref{ex1} shows a pictorial representation of a history $H1:r_1(x, 0) w_2(x, 10) w_2(y, 10) c_2 r_1(y, 0) c_1$. It can be seen that in $H1$, $c_0 = \valw{r_1(x,0)}{H1} = \lwrite{r_1(x,0)}{H1}$. Hence, $r_1(x,0)$ is \legal. But $c_0 = \valw{r_1(y,0)}{H1} \neq c_1 = \lwrite{r_1(y,0)}{H1}$. Thus, $r_1(y,0)$ is \valid{} but not \legal{}

We define a history $H$ as \textit{\multv} if it is \valid{} but \textbf{not} \legal. If a hisory $H$ is \multv{}, then there is at least one read, say $r_i(x)$ in $H$ that is \valid{} but not \legal. The history $H1$ is \multv. Along the same lines, we say that a STM implementation is \multv{} if it exports atleast one history that is \multv. 

\vspace{1mm}
\noindent
\textit{Opacity.} We say that two histories $H$ and $H'$ are \emph{equivalent} if they have the same set of events. Now a history $H$ is said to be \textit{opaque} \cite{GuerKap:2008:PPoPP,tm-book} if $H$ is \valid{} and there exists a \tseq{} legal history $S$ such that (1) $S$ is equivalent to $\overline{H}$ and (2) $S$ respects $\prec_{H}^{RT}$, i.e $\prec_{H}^{RT} \subseteq \prec_{S}^{RT}$. By requiring $S$ being equivalent to $\overline{H}$, opacity treats all the incomplete transactions as aborted. 

\cmnt {
Along the same lines, a \valid{} history $H$ is said to be
\textit{strictly serializable} if $\shist{\comm(H)}{H}$ is opaque.
Thus, unlike opacity, strict serializability does not include aborted
transactions in the global serialization order.

\vspace{1mm}
\noindent
\textit{Conflict Opacity:} A history $H$ is said to be Conflict Opaque(CO), if $\exists S$(serial history): (1) Conf(H) $\approx$ Conf(S), (2) S is legal and (3) $\prec_H^{RT} \approx \prec_S^{RT}$.

\vspace{1mm}
\noindent
\textit{Live Set:} For a given transaction $T_i$, when $T_i$ commits/aborts, at that time all the transactions which are still alive called live set of $T_i$.
}

%% file: mcn.tex
\section{New Conflict Notion for Multi-Version Systems}
\label{sec:mcn}

\subsection{Motivation for a New Conflict Notion}
\label{subsec:motive}

It is not clear if checking whether a history is opaque or can be performed in polynomial time. Checking for membership of \textit{multi-version view-serializability (\mvsr)} \cite[chap. 3]{WeiVoss:2002:Morg}, the correctness criterion for databases, has been proved to be NP-Complete \cite{PapadKanel:1984:MultVer:TDS}. We believe that the membership of opacity, similar to \mvsr, can not be efficiently verified.

In databases a sub-class of \mvsr, \textit{conflict-serializability (\csr)} \cite[chap. 3]{WeiVoss:2002:Morg} has been identified, whose membership can be efficiently verified. As a result, \csr{} is the commonly used correctness criterion in databases since it can be efficiently verified. In fact all known single-version schedulers known for databases are a subset of \csr. Similarly, using the notion of conflicts, a sub-class of opacity, \textit{conflict-opacity (\coopty)} can be designed whose membership can be verified in polynomial time. Further, using the verification mechanism, an efficient STM implementation can be designed that is permissive w.r.t \coopty{} \cite{KuzSat:Corr:2012}.

By storing multiple versions for each \tobj, multi-version STMs provide more concurrency than single-version STMs. But the main drawback of \coopty{} is that it does not admit histories that are \multv. In other words, the set of histories exported by any STM implementation that uses multiple versions is not a subset of \coopty. Thus it can be seen that the set \coopty{} does not take advantage of the concurrency provided by using multiple versions. As a result, it is not clear if a multi-version STM implementation can be developed that is permissive w.r.t some sub-class of opacity. In the rest of this sub-section, we formally define \coopty{} and prove this result. 

The following definitions and proofs are coming directly from \cite{KuzSat:Corr:2012}. We define \coopty{} using \textit{conflict order}~\cite[Chap. 3]{WeiVoss:2002:Morg} as: For two transactions $T_k$ and $T_m$ in $\txns(H)$, we say that \emph{$T_k$  precedes $T_m$ in conflict order}, denoted $T_k \prec_H^{CO} T_m$, if (w-w order) $\tryc_k(C)<_H \tryc_m(C)$ and $Wset(T_k) \cap Wset(T_m) \neq\emptyset$, (w-r order) $\tryc_k(C)<_H r_m(x,v)$, $x \in Wset(T_k)$ and $v \neq A$, or (r-w order) $r_k(x,v)<_H \tryc_m(C)$, $x\in Wset(T_m)$ and $v \neq A$.

Thus, it can be seen that the conflict order is defined only on \op{s} that have successfully executed. Using conflict order, \coopty{} is defined as follows: . 

\begin{definition}
\label{def:coop1}
A history $H$ is said to be \emph{conflict opaque} or \emph{\coop} if $H$ is \valid{} 
and there exists a t-sequential legal history $S$ such that (1) $S$ is equivalent to
$\overline{H}$ and (2) $S$ respects $\prec_{H}^{RT}$ and $\prec_{H}^{CO}$. 
\end{definition}
Readers familiar with the databases literature may notice \coopty{} is analogous to the \emph{order commit conflict serializability} (OCSR)~\cite{WeiVoss:2002:Morg}.

\begin{lemma}
\label{lem:co-eq}
Consider two histories $H1$ and $H2$ such that $H1$ is equivalent to
$H2$ and $H1$ respects conflict order of $H2$, i.e., $\prec_{H1}^{CO} \subseteq \prec_{H2}^{CO}$. Then, $\prec_{H1}^{CO} = \prec_{H2}^{CO}$. 
\end{lemma}

\begin{proof}
Here, we have that $\prec_{H1}^{CO} \subseteq \prec_{H2}^{CO}$. In order to prove $\prec_{H1}^{CO} = \prec_{H2}^{CO}$, we have to show that $\prec_{H2}^{CO} \subseteq \prec_{H1}^{CO}$. We prove this using contradiction. Consider two events $p,q$ belonging to transaction $T1,T2$ respectively in $H2$ such that $(p,q) \in \prec_{H2}^{CO}$ but $(p,q) \notin \prec_{H1}^{CO}$. Since the events of $H2$ and $H1$ are same, these events are also in $H1$. This implies that the events $p, q$ are also related by $CO$ in $H1$. Thus, we have that either $(p,q) \in \prec_{H1}^{CO}$  or $(q,p) \in \prec_{H1}^{CO}$. But from our assumption, we get that the former is not possible. Hence, we get that $(q,p) \in \prec_{H1}^{CO} \Rightarrow (q,p) \in \prec_{H2}^{CO}$. But we already have that $(p,q) \in \prec_{H2}^{CO}$. This is a contradiction. 
\end{proof}

\begin{lemma}
\label{lem:eqv-legal}
Let $H1$ and $H2$ be equivalent histories such that 
$\prec_{H1}^{CO} = \prec_{H2}^{CO}$. 
Then $H1$ is \legal{} iff $H2$ is \legal. 
\end{lemma}

\begin{proof}
It is enough to prove the `if' case, and the `only if' case will follow from symmetry of the argument. Suppose that $H1$ is \legal{}. By contradiction, assume that $H2$ is not \legal, i.e., there is a read \op{} $r_j(x,v)$ (of transaction $T_j$) in $H2$ with its \lastw{} as $c_k$ (of transaction $T_k$) and $T_k$ writes $u \neq v$ to $x$, i.e. $w_k(x, u) \in \evts{T_k}$.  Let $r_j(x,v)$'s \lastw{} in $H1$ be $c_i$ of $T_i$. Since $H1$ is legal, $T_i$ writes $v$ to $x$, i.e. $w_i(x, v) \in \evts{T_i}$. 

Since $\evts{H1} = \evts{H2}$, we get that $c_i$ is also in $H2$, and
$c_k$ is also in $H1$. 
As $\prec_{H1}^{CO} = \prec_{H2}^{CO}$,
we get $c_i <_{H2} r_j(x, v)$ and $c_k <_{H1} r_j(x, v)$. 

Since $c_i$ is the \lastw{} of $r_j(x,v)$ in $H1$ we derive that 
$c_k <_{H1} c_i$ and, thus, $c_k <_{H2} c_i <_{H2} r_j(x, v)$.
But this contradicts the assumption that $c_k$ is the \lastw{} of
$r_j(x,v)$ in $H2$.
Hence, $H2$ is legal. \qed
\end{proof}

We now prove that if a history is \multv, then it is not in \coopty. 

\begin{lemma}
\label{lem:multi-co}
If a history $H$ is \multv{} then $H$ is not in \coopty. Formally, $\langle (H \text{is \multv}) \implies (H \notin \text{\coopty}) \rangle$. 
\end{lemma}

\begin{proof}
We prove this using contradiction. Assume that $H$ is \multv{}, i.e. $H$ is \valid{} but not \legal. But suppose that $H$ is in \coopty. From the definition of \coopty, we get that there exists a sequential and legal history $S$ such that $\prec_{H}^{CO} \subseteq \prec_{S}^{CO}$. From \lemref{co-eq}, we get that $\prec_{H}^{CO} = \prec_{S}^{CO}$. Combining this with \lemref{eqv-legal} and the assumption that $H$ is not legal, we get that $S$ is not legal. But this contradicts out assumption that $S$ legal. Hence, $H$ is not in \coopty. \qed
\end{proof}

Having shown the shortcoming of conflicts, we now define a new conflict notion in the next sub-section that will accommodate \multv{} histories as well.

\subsection{Multi-Version Conflict Definition}
\label{subsec:mvc-defn}

We define a few notations to describe the conflict notion. Consider a history $H$. For a read $r_i(x,v)$ in $H$, we define $r_i$'s \emph{\vwrite}, formally $\valw{r_i}{H}$, as the commit \op{} $c_j$ belonging to the transaction $T_j$ that occurs before $r_x$ in $H$ and writes $v$ to $x$. If there are multiple such committed transactions that write $v$ to $x$ then the \vwrite{} is the commit \op{} closest to $r_x$. 

\begin{definition}
\label{defn:mvco}
For a history $H$, we define \textit{\mvc} order, denoted as $\prec^{\mvco}_{H}$, between \op{s} of $\overline{H}$ as follows: (a) commit-commit (c-c) order: $c_i \prec^{\mvco}_{H} c_j$ if $c_i <_H c_j$ for two committed transaction $T_i$, $T_j$ and both of them write to $x$. (b) commit-read (c-r) order: Let $r_i(x, v)$ be a read \op{} in $H$ with its \vwrite{} as $c_j$ (belonging to the committed transaction $T_j$). Then for any committed transaction $T_k$ that writes to $x$ and either commits before $T_j$ or is same as $T_j$, formally $(c_k <_H c_j) \lor (c_k = c_j)$ , we define $c_k \prec^{\mvco}_{H} r_i$. (c) read-commit (r-c) order: Let $r_i(x, v)$ be a read \op{} in $H$ with its \lastw{} as $c_j$ (belonging to the committed transaction $T_j$). Then for any committed transaction $T_k$ that writes to $x$ and commits after $T_j$, i.e. $c_j <_H c_k$, we define $r_i \prec^{\mvco}_{H} c_k$.
\end{definition}

Observe that the \mvc{} order is defined on the \op{s} of $\overline{H}$ and not $H$. The set of conflicts in $H1$ are: $[\text{c-r}: (c_0, r_1(x,0)), (c_0, r_1(y,0))], [\text{r-c}: (r_1(x,0), c_2), (r_1(y,0), c_1)], [\text{c-c}: (c_0, c_2)]$. 

We say that a history $H'$ \textit{satisfies} the \mvc{} order of a history $H$, $\prec_H^{\mvco}$ if: (1) the events of $H'$ are same as $\overline{H}$, i.e. $H'$ is equivalent to $\overline{H}$. (2) For any two \op{s} $op_i$ and $op_j$, $op_i \prec^{\mvco}_{H} op_j$ implies $op_i <_{H'} op_j$. We denote this by $H' \vdash \prec_H^{\mvco}$. Otherwise, we say that $H'$ does not satisfy $\prec^{\mvco}_{H}$ and denote it as $H' \nvdash \prec_H^{\mvco}$.

Note that for any history $H$ that is \multv{}, $H$ does not satisfy its own \mvc{} order $\prec^{\mvco}_{H}$. For instance the \multv{} order in history $H1$ consists of the pair: $(r_1(y, 0), c_2)$. But $c_2$ occurs before $r_1(y, 0)$ in $H1$. We formally prove this property using the following lemmas.

\begin{lemma}
\label{lem:satisfy-valid}
Consider a (possibly \multv) \valid{} history $H$. Let $H'$ be a history which satisfies $\prec^{\mvco}_{H}$. Then $H'$ is \valid{} and $\prec^{\mvco}_{H'} = \prec^{\mvco}_{H}$. Formally, $\langle (H \text{ is \valid}) \land (H' \vdash \prec^{\mvco}_{H}) \implies (H' \text{ is \valid}) \land (\prec^{\mvco}_{H'} = \prec^{\mvco}_{H}) \rangle $.
\end{lemma}

\begin{proof}
Here, we have that $H$ is \valid{} and $H'$ satisfies $\prec^{\mvco}_{H}$. Thus all the events of $H$ and $H'$ are the same. The definition of satisfaction says that the events of $H'$ are ordered according to \mvc{} order of $H$. Thus, it can be verified that all the \mvc{} orders of both histories are the same, i.e. $\prec^{\mvco}_{H'} = \prec^{\mvco}_{H}$. Since $H$ is \valid{} and $H'$ has the same c-r \mvc{} order as $H$, the \vwrite{} of all the read \op{s} in $H'$ occur before the corresponding reads in $H'$. Hence $H'$ is \valid{} as well. \qed
\end{proof}

\begin{lemma}
\label{lem:mvco-legal}
Consider a (possibly \multv) \valid{} history $H$. Let $H'$ be a \valid{} history (which could be same as $H$). If $H'$ satisfies $\prec^{\mvco}_{H}$ then $H'$ is \legal. Formally, $\langle (H' \text{ is \valid}) \land (H' \vdash \prec^{\mvco}_{H}) \implies (H' \text{ is \legal}) \rangle $.
\end{lemma}

\begin{proof}
Assume that $H'$ is not \legal. Hence there exists a read \op{}, say $r_i(x, v)$, in $\evts{H'}$ that is not \legal. This implies that \lastw{} of $r_i$ is not the same as its \vwrite. Let $c_l = \lwrite{r_i}{H'} \neq \valw{r_i}{H'} = c_v$. Let $w_l(x,u) \in \evts{T_l}$ and $w_v(x,v) \in \evts{T_v}$. Since $H'$ is \valid{} and $c_l$ is the \lastw{} of $r_i$, we get the following order of the events: 

\begin{equation}
\label{eq:v-l-order}
c_v <_{H'} c_l <_{H'} r_i 
\end{equation}

Since $H' \vdash \prec^{\mvco}_{H}$ and $H$ is \valid, from \lemref{satisfy-valid} we get that $\prec^{\mvco}_{H} = \prec^{\mvco}_{H'}$. Hence, $H'$ must satisfy its own \mvc{} order, $\prec^{\mvco}_{H'}$. Now combining \eqnref{v-l-order} with the definition of r-w \mvc{} order, we get that $r_i \prec^{\mvco}_{H'} (= \prec^{\mvco}_{H}) c_l$. But since $H' \vdash \prec^{\mvco}_{H'}$, we get that $r_i <_{H'} c_l$. But this contradicts \eqnref{v-l-order}. Hence, $r_i$ must be legal which in turn implies that $H'$ must be legal. \qed  
\end{proof}

\noindent Using this lemma, we get the following corollary, 

\begin{corollary}
\label{cor:mltv-mvc}
Consider a (possibly \multv) \valid{} history $H$. Let $H'$ be a \multv{} history (which could be same as $H$). Then, $H'$ does not satisfy $\prec^{\mvco}_{H}$. Formally, $\langle (H' \text{ is \multv}) \implies (H' \nvdash \prec^{\mvco}_{H}) \rangle $.
\end{corollary}

\begin{proof}
We are given that $H'$ is \multv. This implies that $H'$ is not \legal. From the contrapositive of \lemref{mvco-legal}, we get that $H' \nvdash \prec^{\mvco}_{H}$. \qed
\end{proof}

\section{Multi-Version Conflict Opacity}
\label{sec:mvco}

We now illustrate the usefulness of the conflict notion by defining another subset of opacity \emph{\mvopty} which is larger than \coopty. We formally define it as follows (along the same lines as \coopty): 

\begin{definition}
\label{def:mvcop}
A history $H$ is said to be \emph{multi-version conflict opaque} or \emph{\mvop} if $H$ is \valid{} and there exists a t-sequential history $S$ such that (1) $S$ is equivalent to $\overline{H}$ and (2) $S$ respects $\prec_{H}^{RT}$ and $S$ satisfies $\prec_{H}^{\mvco}$. 
\end{definition}

It can be seen that the history $H1$ is \mvop, with the \mvc{} equivalent \tseq{} history being: $T1 T2$.  Please note that we don't restrict $S$ to be \legal{} in the definition. But it turns out that if $H$ is \mvop{} then $S$ is automatically \legal{} as shown in the following lemma.

\begin{theorem}
\label{thm:mvcop-op}
If a history $H$ is \mvop, then it is also \opq. 
\end{theorem}

\begin{proof}
Since $H$ is \mvop, it follows that there exists a \tseq{} history $S$ such that (1) $S$ is equivalent to $\overline{H}$ and (2) $S$ respects $\prec_{H}^{RT}$ and $S$ satisfies $\prec_{H}^{\mvco}$. In order to prove that $H$ is \opq, it is sufficient to show that $S$ is \legal. Since $S$ satisfies $\prec_{H}^{\mvco}$, from \lemref{mvco-legal} we get that $S$ is \legal. Hence, $H$ is \opq{} as well. \qed 
\end{proof}

Thus, this lemma shows that \mvopty{} is a subset of \opty. Actually, \mvopty{} is a strict subset of \opty. Consider the history $H2 = r_1(x,0) r_2(z,0) r_3(z,0) w_1(x, 5) c_1 r_2(x, 5) w_2(x, 10) w_2(y, 15) c_2 \\ 
r_3(x, 5) w_3(y, 25) c_3$. The set of \mvc{s} in $H2$ are (ignoring the conflicts with $C_0$): $[\text{c-r}: (c_1, r_2(x,5)), (c_1, r_3(x,5))], [\text{r-c}: (r_3(x,5), c_2)], [\text{c-c}: (c_1, c_2), (c_2, c_3)]$. It can be verified that $H2$ is \opq{} with the equivalent \tseq{} history being $T_1 T_3 T_2$. But there is no \mvc{} equivalent \tseq{} history. This is because of the conflicts: $(r_3(x,5), c_2), (c_2, c_3)$. Hence, $H2$ is not \mvop.

Next, we will relate the classes \coopty{} and \mvopty. In the following lemma, we show that \coopty{} is a subset of \mvopty. 

\begin{theorem}
\label{thm:cop-mvcop}
If a history $H$ is \coop, then it is also \mvop. 
\end{theorem}

\begin{proof}
Since $H$ is \coop, we get that there exists an equivalent \legal{} \tseq{} history $S$ that respects the real-time and conflict orders of $H$. Thus if we show that $S$ satisfies \mvc{} order of $H$, it then implies that $H$ is also \mvop. Since $S$ is \legal, it turns out that the conflicts and \mvc{s} are the same. To show this, let us analyse each conflict order:
\begin{itemize}

\item c-c order: If two \op{s} are in c-c conflict, then by definition they are also ordered by the c-c \mvc. 

\item c-r order: Consider the two \op{s}, say $c_j$ and $r_i$ that are in conflict (due to a \tobj{} $x$). Hence, we have that $c_k <_H r_i$. Let $c_j = \valw{r_i}{H}$. Since, $S$ is \legal, either $c_k = c_j$ or $c_k <_H c_j$. In either case, we get that $c_k \prec_H^{\mvco} r_i$.

\item r-c order: Consider the two \op{s}, say $c_j$ and $r_i$ that are in conflict (due to a \tobj{} $x$). Hence, we have that $r_i <_H c_k$. Let $c_j = \valw{r_i}{H}$. Since, $S$ is \legal, $c_j <_H r_i <_H c_k$. Thus in this case also we get that $r_i  \prec_H^{\mvco} c_k$.

\end{itemize}
Thus in all the three cases, conflict among the \op{s} in $S$ also reults in \mvc{} among these \op{s}. Hence, $S$ satisfies the \mvc{} order. \qed
\end{proof}

This lemma shows that \coopty{} is a subset of \mvopty. The history $H1$ is \mvop{} but not in \coop. Hence, \coopty{} is a strict subset of \mvopty. \figref{ex2} shows the relation between the various classes. 

\begin{figure}[tbph]
\centerline{\scalebox{0.5}{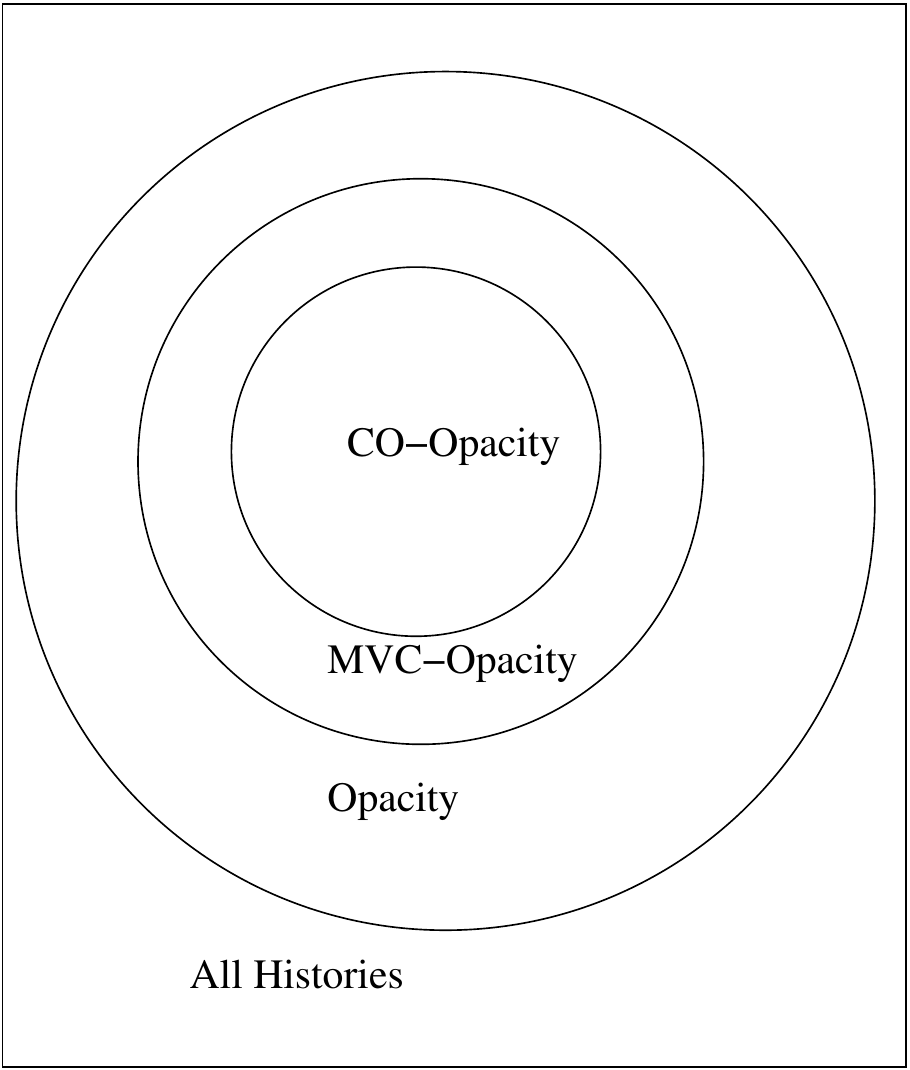}}
\caption{Relation between the various classes}
\label{fig:ex2}
\end{figure}

\subsection{Graph Characterization of MVC-Opacity}
\label{subsec:graph}

In this section, we will describe graph characterization of \mvopty. This characterization will enable us to verify its membership in polynomial time. 

Given a history $H$, we construct a \textit{multi-version conflict graph}, $\mvcg{H} = (V,E)$ as follows:  (1) $V=\txns(H)$, the set of transactions in $H$ (2) an edge $(T_i,T_j)$ is added to $E$ whenever 
\begin{itemize}
\item[2.1] real-time edges: If $T_i$ precedes $T_j$ in $H$
\item[2.2] \mvco{} edges: If $T_i$ contains an \op{} $p_i$ and $T_j$ contains $p_j$ such that $p_i \prec_{H}^{\mvco} p_j$.
\end{itemize}

\noindent Based on the \mvg{} construction, we have the following graph characterization for \mvopty. 

\begin{theorem}
\label{thm:graph}
A \valid{} history $H$ is \mvop{} iff $\mvcg{H}$ is acyclic. 
\end{theorem}

\noindent \figref{mvcgraphs} shows the \mvg{s} for the histories $H1$ and $H2$. 

\begin{figure}[tbph]
\centerline{\scalebox{0.5}{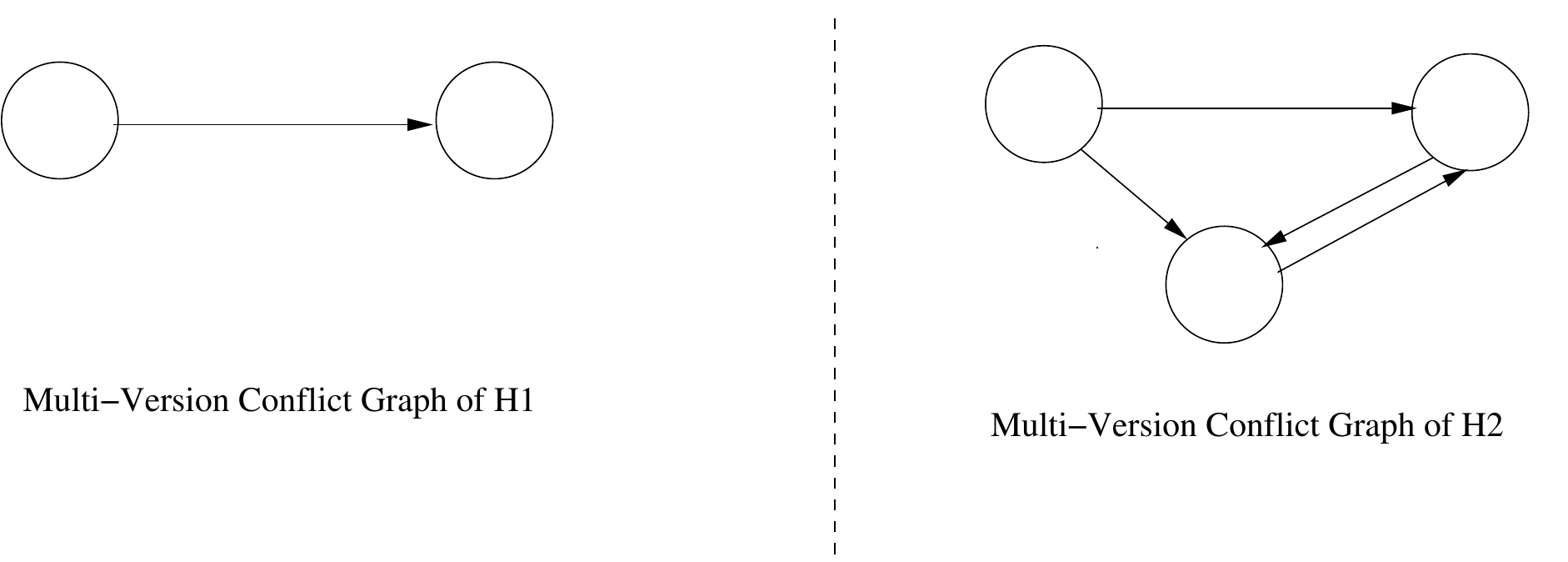}}
\caption{\mvg{s} of $H1$ and $H2$}
\label{fig:mvcgraphs}
\end{figure}

%% file: ex2.pdf_t
\begin{picture}(0,0)%
\includegraphics{ex2.pdf}%
\end{picture}%
\setlength{\unitlength}{4144sp}%
\begingroup\makeatletter\ifx\SetFigFont\undefined%
\gdef\SetFigFont#1#2#3#4#5{%
  \reset@font\fontsize{#1}{#2pt}%
  \fontfamily{#3}\fontseries{#4}\fontshape{#5}%
  \selectfont}%
\fi\endgroup%
\begin{picture}(4155,4884)(844,-4348)
\end{picture}%

%% file: mvcgs.pdf_t
\begin{picture}(0,0)%
\includegraphics{mvcgs.pdf}%
\end{picture}%
\setlength{\unitlength}{4144sp}%
\begingroup\makeatletter\ifx\SetFigFont\undefined%
\gdef\SetFigFont#1#2#3#4#5{%
  \reset@font\fontsize{#1}{#2pt}%
  \fontfamily{#3}\fontseries{#4}\fontshape{#5}%
  \selectfont}%
\fi\endgroup%
\begin{picture}(8589,3072)(1135,-3448)
\put(1171,-1681){\makebox(0,0)[lb]{\smash{{\SetFigFont{14}{16.8}{\rmdefault}{\mddefault}{\updefault}{\color[rgb]{0,0,0} $T_1$}%
}}}}
\put(3556,-1636){\makebox(0,0)[lb]{\smash{{\SetFigFont{14}{16.8}{\rmdefault}{\mddefault}{\updefault}{\color[rgb]{0,0,0} $T_2$}%
}}}}
\put(8911,-511){\makebox(0,0)[lb]{\smash{{\SetFigFont{14}{16.8}{\rmdefault}{\mddefault}{\updefault}{\color[rgb]{0,0,0} $T_2$}%
}}}}
\put(6526,-556){\makebox(0,0)[lb]{\smash{{\SetFigFont{14}{16.8}{\rmdefault}{\mddefault}{\updefault}{\color[rgb]{0,0,0} $T_1$}%
}}}}
\put(7561,-2491){\makebox(0,0)[lb]{\smash{{\SetFigFont{14}{16.8}{\rmdefault}{\mddefault}{\updefault}{\color[rgb]{0,0,0} $T_3$}%
}}}}
\end{picture}%